\documentclass[aps,pra,twocolumn,dvips,nofootinbib]{revtex4-2}
\usepackage{blindtext}
\usepackage{babel}
\usepackage[linktocpage=true, colorlinks=true, linkcolor=red, urlcolor=blue, citecolor=blue]{hyperref}
\usepackage[utf8]{inputenc}
\usepackage{pslatex}
\usepackage{setspace}
\usepackage{bbold}
\usepackage{amsthm}
\usepackage{amsmath}
\usepackage{amssymb}
\usepackage{amsfonts}
\usepackage{mathtools}
\usepackage{paralist, tabularx}
\usepackage{bm}
\DeclareMathOperator{\Tr}{Tr}

\usepackage{indentfirst}
\usepackage{mathrsfs}
\usepackage{braket}
\usepackage{graphicx}
\usepackage[margin=2cm]{geometry}
\usepackage[colorinlistoftodos]{todonotes}
\theoremstyle{definition}
\newtheorem{definition}{Definition}[section]

\newtheorem{proposition}{Proposition}
\newtheorem{result}{Result}[]

\newcommand{\blk}{\color{black}}
\newcommand{\vini}{\color{black}}

\usepackage[export]{adjustbox}
\usepackage{tikz}
\usetikzlibrary{decorations.markings}
\usetikzlibrary{shapes.geometric}
\pgfdeclarelayer{edgelayer}
\pgfdeclarelayer{nodelayer}
\pgfsetlayers{edgelayer,nodelayer,main}
\tikzstyle{none}=[inner sep=0pt]
\tikzstyle{simple}=[-,draw=Black,line width=2.000]
\tikzstyle{zbasis}=[fill=white, draw=black, shape=circle, minimum size=2mm]
\tikzstyle{xbasis}=[fill={rgb,255: red,191; green,191; blue,191}, draw=black, shape=circle]
\tikzstyle{hadamard}=[circle,fill=blue!50, draw=black,inner sep=1.5pt]
\tikzstyle{diamond}=[fill=white, draw=black, shape=diamond]
\pgfdeclarelayer{background}
\pgfsetlayers{background,edgelayer,nodelayer}

\usepackage{soul}

\begin{document}
	\title{Contextuality with vanishing coherence and maximal robustness to dephasing}
	\author{Vinicius P. Rossi}
	\email{prettirossi.vinicius@gmail.com}
	\author{David Schmid}
	\author{John H. Selby}
	\author{Ana Belén Sainz}
	\affiliation{International Centre for Theory of Quantum Technologies, University of Gda{\'n}sk, 80-309 Gda\'nsk, Poland}
	
	\begin{abstract}
		Generalized contextuality is a resource for a wide range of communication and information processing protocols. However, contextuality is not possible without coherence, and so can be destroyed by dephasing noise. Here, we explore the robustness of contextuality to partially dephasing noise in a scenario related to state discrimination (for which contextuality is a resource). We find that a vanishing amount of coherence is sufficient to demonstrate the failure of noncontextuality in this scenario, and we give a proof of contextuality that is robust to arbitrary amounts of partially dephasing noise. This is in stark contrast to partially depolarizing noise, which is always sufficient to destroy contextuality.
	\end{abstract}
	
	\maketitle
	
	Understanding what is nonclassical about quantum theory is crucial for determining which tasks can be optimally performed with quantum resources. One quantum resource that is useful in many tasks \vini within \blk computation \cite{COMPUTATION5}, communication \cite{COMMUNICATION1}, information processing \cite{PROCESSING1,PROCESSING2,PROCESSING3,PROCESSING5}, metrology \cite{METROLOGY1}, cloning \cite{CLONING1}, and state discrimination \cite{DAVID1, SD1,SD2,SD3}, is generalized contextuality~\cite{SPEKKENS} (henceforth referred to simply as `contextuality'). 
	
	A given experiment is said to be a proof of contextuality when its statistics are incompatible with the existence of a noncontextual ontological model, i.e., models wherein one's ontology is a set (of classical states), dynamics are represented as functions, where inferences are made using Bayesian probability theory and Boolean logic, and where a methodological version of the assumption of \emph{Leibnizianity} \cite{LEIBNIZ} is satisfied. This assumption stipulates that the explanation for procedures being indiscernible at the operational level is that they are also indiscernible at the ontological level \cite{OMELETTE,INTERFERENCE}. This notion of nonclassicality -- the nonexistence of a noncontextual ontological model -- was proven to be equivalent to other notions of nonclassicality, such as the nonexistence of a quasiprobability representation in quantum optics \cite{NEGATIVITY,STRUCTURE-THEO},  and the nonexistence of a simplex embedding in generalized probabilistic theories \cite{SIMPLEX}.  Furthermore, this notion of nonclassicality is closely related to the notions arising in the study of quantum Darwinism \cite{BALDIJAO}, macrorealism \cite{MACROREALISM}, Bell scenarios \cite{SEER,INEQUALITIES,VICKY}, and the detection of anomalous weak values \cite{WEAK}. In our view, generalized contextuality is our most well-motivated notion of nonclassicality. 
	
	Of particular interest to this work is the aforementioned notion of simplex embeddability. This simple geometric characterization of the notion of noncontextuality within the framework of GPTs has been useful for exploring the relationship between contextuality and incompatibility \cite{INCOMPATIBILITY, FRAGMENTS}. It has also been employed in the development of an open-source code for testing whether a given prepare-and-measure scenario constitutes a proof of contextuality, and, moreover, for providing a quantification of how robust to depolarizing noise this proof is \cite{CODE}. We will leverage this tool here in our study of the relationship between contextuality and coherence.
	
	It is well known that contextuality is always destroyed by partial (but sufficiently large) depolarizing noise \footnote{This fact has been noted in particular scenarios \cite{DAVID1,RAVIROBUST,IMAN}, and in general scenarios it follows immediately from simplex-embeddability.}. However, the question of how robust contextuality is to dephasing noise has not previously been studied.  Note that the existence of coherence does not immediately imply contextuality, since it is present in epistemically restricted theories \cite{INTERFERENCE,TOYTHEORY,TOY2,TOY3} for which noncontextual ontological models are known to exist. On the other hand, contextuality is not possible without coherence, a fact that we prove explicitly in Appendix \ref{app2}. However, this leaves open the question of how contextuality is affected by {\em partial} dephasing noise. This question is of particular importance given that decoherence theory \cite{DECOHERENCE1,DECOHERENCE2} shows that dephasing noise arises in generic open system dynamics.
	
	In this work, we show that there are proofs of contextuality that can be obtained with \emph{any} non-zero amount of coherence. We then modify the open-source linear program from Ref.~\cite{CODE} and use this to investigate the robustness of contextuality to the action of dephasing noise with respect to a fixed basis in a collection of prepare-and-measure scenarios. Finally, we find a proof of contextuality that is maximally robust to dephasing noise, in the sense that the experiment remains a proof of contextuality for any amount of decoherence apart from total decoherence.
	
	\section{Preliminary notions}\label{sec1} 
	
	The broad range of experiments we are interested in investigating consists of those that can be thought of as preparing a system in a laboratory in a variety of different ways and probing it with a variety of measurements. Such experiments are known as  \emph{prepare-and-measure scenarios}. These can be studied from a theory-agnostic viewpoint, where only a minimal set of \textit{operational} elements (i.e., properties or objects that are manifestly observable) are used to describe it. In such situations, one can analyze the experimental scenario without making any assumptions about the nature of the system in question, e.g., what its intrinsic properties are or how it behaves. Rather, one simply focuses on (i) the classical labels of the ways in which one may prepare this system ($P\in \mathcal{P}$), (ii) the classical labels of the measurement procedures one may perform ($M\in \mathcal{M}$), (iii) the classical labels of the outcomes $k\in K$ of the measurement procedures, (iv) the resulting statistics $\{p(k|M,P)\}_{ [k|M]\in K\times \mathcal{M}, P\in\mathcal{P}}$ of the experiment\vini, and (v) the operational equivalences  between different preparation or measurement outcome procedures, denoted denoted $\simeq$ and defined in the next paragraph\blk. The tuple $(\mathcal{P},\mathcal{M}, K, p\vini,\simeq\blk)$ that captures the information analyzed in such a prepare-and-measure scenario is often referred to as an \emph{operational scenario}. 
	
	An \emph{operational theory} is the set of all possible
	realisable preparations, measurements and outcomes and their respective statistics \vini -- i.e., it is the maximal operational scenario for a given system. \blk
	Often a given operational theory allows for a kind of equivalence between different procedures: sometimes two preparation procedures $P$ and $P'$ yield the same statistics for any possible measurement outcomes, or two measurement outcomes $[k|M],\, [k'|M']$ do so for any possible preparation. When this happens, we say that $P$ and $P'$ (resp.~$[k|M]$ and $[k'|M']$) are \emph{inferentially equivalent}, denoted by $P\simeq P'$ (resp.~$[k|M]\simeq [k'|M']$). \vini Here, it is crucial to assess such equivalences with respect to the full operational theory, and not only relative to the specific preparations or measurement outcomes in the specific operational scenario under investigation.  \blk \vini That is, we define the operational equivalence relation for an operational scenario as the one that it inherits from the operational theory in which it lives. \blk
	
	Since inferentially-equivalent processes cannot be distinguished by the operational predictions $p$ they generate, it is often useful to discard this {\em context} information by identifying inferentially-equivalent processes with a single representative of the group. This operation is termed \textit{quotienting} \cite{CHIRIBELLA} and provides the way to construct a so-called \emph{generalized probabilistic theory} (GPT) \cite{GPT0,GPT1,GPT2} for a corresponding operational theory. See Appendix \ref{app3} for a concrete definition of a GPT.
	
	An {\em ontological model} seeks to explain the observed statistics in one's scenario by associating (i) the system to some set of ontic states $\Lambda$, (ii) preparations $P$ to \emph{epistemic states}, $\mu_P$, which are probability distributions over $\Lambda$, and (iii) measurement-outcome pairs $[k|M]$ in the operational theory to \emph{response functions} $\xi_{k|M}$ over $\Lambda$, such that $p(k|M,P)=\sum_{\lambda \in \Lambda} \xi_{k|M}(\lambda)\mu_P(\lambda)$. An ontological model is said to be noncontextual if inferentially-equivalent preparations are mapped to the same epistemic state, and inferentially-equivalent measurement-outcome pairs are mapped to the same response function\footnote{\vini It is for this reason that inferential equivalences should be assessed relative to the entire scope of possible procedures in the theory rather than those of the scenario -- it does not make sense to impose a constraint on an ontological description which is contingent on what we happen to have chosen to do in a given experiment.\blk}. For the purpose of this paper, it suffices to know that the notion of a noncontextual ontological model for the operational theory has an equivalent characterization at the level of the GPT related to it via quotienting. That is, a GPT is associated to an operational theory that is noncontextual if and only if the GPT is simplex-embeddable \cite{SIMPLEX}. Intuitively, such GPTs have a state space that fits inside a simplex, and an effect space that fits inside the dual to the simplex. We define simplex-embeddability formally in Appendix \ref{app3}. 
	
	The existence of a simplex-embedding can be tested using the linear program  introduced in Ref.~\cite{CODE}.
	In the case of quantum theory (which is the case we study here) the linear program simply takes as input a set of density matrices (representing the preparations) and a set of POVM elements (representing the measurement-outcome pairs), and checks whether or not these are simplex-embeddable, and consequently, whether the quantum scenario admits of a noncontextual ontological model. Furthermore, in the case that the code fails to find a simplex embedding, it computes how much depolarizing noise must be added to the input states (or equivalently, measurements) such that a simplex embedding is found. In this work, we are interested in studying quantum prepare-and-measure scenarios under the action of \emph{dephasing} rather than depolarizing noise, so we modify the code from Ref.~\cite{CODE} to estimate the robustness to dephasing rather than depolarization. A summary of how the code works and of our modifications to it is given in Appendix \ref{app3}.
	
	\section{Proof of contextuality with vanishing coherence}\label{sec2}
	
	Inspired by Ref.~\cite{DAVID1} -- which demonstrates that contextuality is a resource powering an advantage for minimum-error state discrimination (MESD) -- we focus on a prepare-and-measure scenario constructed from the MESD scenario. Our scenarios of interest consist of four preparations $\{P_\psi, P_{\bar{\psi}},P_\phi,P_{\bar{\phi}}\}$ and three binary measurements $\{M_\psi,M_\phi,M_g\}$. The preparations consist of pure states $\ket{\psi}$, $\ket{\phi}$ of a qubit system, and their orthogonal counterparts. That is, $P_\psi \rightarrow \ket{\psi}\bra{\psi}$ and $P_{\bar{\psi}} \rightarrow \ket{{\bar{\psi}}}\bra{{\bar{\psi}}}$, with $\langle \psi \vert {\bar{\psi}} \rangle = 0$ (and similarly for $P_\phi$). Measurements $M_\psi$ and $M_\phi$ are simply projections onto $\ket{\psi}$ and $\ket{\phi}$, respectively, while $M_g$ is the Helstrom measurement, comprised of projectors onto the basis that straddles $\ket{\psi}$ and $\ket{\phi}$ \cite{HALSTROM}. As  all these preparations and measurements lie within a two-dimensional slice of the Bloch sphere, we can, without loss of generality, take this slice to be the $ZX$ plane of the Bloch sphere. Furthermore, we fix our system of coordinates such that the projectors $E_{g_\psi}$ and $E_{g_\phi}$ associated with the measurement $M_g$ lie aligned to the $Z$ axis. The preparations and measurements in the scenario can be parameterized by the angle $\theta\in\left[0,\frac{\pi}{2}\right]$ between any of the preparations and the $Z$ axis, as shown in Figure \ref{Fig1}.
	We will consider dephasing relative to the $Z$ axis.
	\\

	\begin{figure}[b]
		\centering
		\begin{tabular}{cc}\vspace{5pt}
			\adjustbox{scale=0.8,valign=m}{\begin{tikzpicture}
	\begin{pgfonlayer}{nodelayer}
		\node [style=none] (0) at (3.925, 3.25) {};
		\node [style=none] (1) at (2.675, 2) {};
		\node [style=none] (2) at (3.925, 0.75) {};
		\node [style=none] (3) at (5.175, 2) {};
		\node [style=hadamard] (4) at (4.825, 2.825) {};
		\node [style=hadamard] (5) at (3.025, 1.175) {};
		\node [style=none] (6) at (5.525, 3.2) {$\ket{\psi}\bra{\psi}$};
		\node [style=hadamard] (7) at (3.025, 2.825) {};
		\node [style=hadamard] (8) at (4.825, 1.175) {};
		\node [style=none] (9) at (2.425, 3.2) {$\ket{\bar{\phi}}\bra{\bar{\phi}}$};
		\node [style=none] (10) at (2.35, 0.75) {$\ket{\bar{\psi}}\bra{\bar{\psi}}$};
		\node [style=none] (11) at (5.45, 0.75) {$\ket{\phi}\bra{\phi}$};
		\node [style=none] (12) at (3.925, 2.45) {};
		\node [style=none] (13) at (4.275, 2.325) {};
		\node [style=none] (14) at (4.175, 2.7) {$\theta$};
		\node [style=none] (15) at (3.925, 1.55) {};
		\node [style=none] (16) at (4.275, 1.675) {};
		\node [style=none] (17) at (4.175, 1.3) {$\theta$};
	\end{pgfonlayer}
	\begin{pgfonlayer}{edgelayer}
		\draw [bend left=45] (1.center) to (0.center);
		\draw [bend left=45] (0.center) to (3.center);
		\draw [bend left=45] (3.center) to (2.center);
		\draw [bend left=45] (2.center) to (1.center);
		\draw (5) to (4);
		\draw (8) to (7);
		\draw [dashed] (0.center) to (2.center);
		\draw [bend left] (12.center) to (13.center);
		\draw [bend right] (15.center) to (16.center);
	\end{pgfonlayer}
\end{tikzpicture}} &\adjustbox{scale=0.8,valign=m}{\begin{tikzpicture}
	\begin{pgfonlayer}{nodelayer}
		\node [style=hadamard] (0) at (3.925, 3.25) {};
		\node [style=hadamard] (1) at (3.925, 0.75) {};
		\node [style=none] (2) at (3.925, 3.675) {$E_{g_\psi}$};
		\node [style=none] (3) at (3.925, 0.25) {$E_{g_\phi}$};
		\node [style=none] (4) at (3.925, 2.45) {};
		\node [style=none] (5) at (4.275, 2.325) {};
		\node [style=none] (6) at (4.175, 2.7) {$\theta$};
		\node [style=none] (7) at (3.925, 3.25) {};
		\node [style=none] (8) at (2.675, 2) {};
		\node [style=none] (9) at (3.925, 0.75) {};
		\node [style=none] (10) at (5.175, 2) {};
		\node [style=hadamard] (11) at (4.825, 2.825) {};
		\node [style=hadamard] (12) at (3.025, 1.175) {};
		\node [style=none] (13) at (5.525, 3.2) {$\ket{\psi}\bra{\psi}$};
		\node [style=hadamard] (14) at (3.025, 2.825) {};
		\node [style=hadamard] (15) at (4.825, 1.175) {};
		\node [style=none] (16) at (2.425, 3.2) {$\ket{\bar{\phi}}\bra{\bar{\phi}}$};
		\node [style=none] (17) at (2.35, 0.75) {$\ket{\bar{\psi}}\bra{\bar{\psi}}$};
		\node [style=none] (18) at (5.45, 0.75) {$\ket{\phi}\bra{\phi}$};
		\node [style=none] (19) at (3.925, 1.55) {};
		\node [style=none] (20) at (4.275, 1.675) {};
		\node [style=none] (21) at (4.175, 1.3) {$\theta$};
	\end{pgfonlayer}
	\begin{pgfonlayer}{edgelayer}
		\draw (1) to (0);
		\draw [bend left] (4.center) to (5.center);
		\draw [bend left=45] (8.center) to (7.center);
		\draw [bend left=45] (7.center) to (10.center);
		\draw [bend left=45] (10.center) to (9.center);
		\draw [bend left=45] (9.center) to (8.center);
		\draw (12) to (11);
		\draw (15) to (14);
		\draw [bend right] (19.center) to (20.center);
	\end{pgfonlayer}
\end{tikzpicture}}
		\end{tabular}
		\caption{Preparation (left) and measurement (right) procedures in the studied scenario, represented on a 2d slice of the Bloch sphere. The vertical axis is taken to be the $Z$ axis.}\label{Fig1}
	\end{figure}
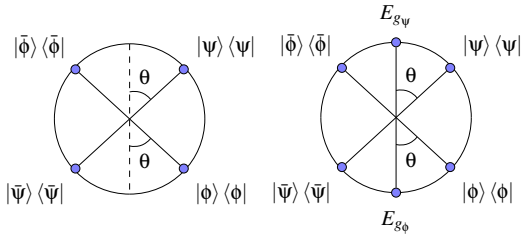
	Notice that the parameter $\theta$ is closely related to the amount of coherence (relative to the $Z$ axis) in the preparation and measurement procedures. If the coherence quantifier, $C$, is the trace distance of the state from the fully dephased version of the state \cite{COHERENCE}, for instance, then we find that
	\begin{equation}\label{quantif}
		C\left(\ket{\psi}\!\!\bra{\psi}\right):=\left\|\ket{\psi}\!\!\bra{\psi}-\sum_{i=0}^1|\!\!\braket{i|\psi}\!\!|^2\ket{i}\!\!\bra{i}\right\|_{1}=\sin\theta,
	\end{equation}
	with the same result for all preparations and measurements (other than $E_{g_\psi}$ and $E_{g_\phi}$ for which the coherence is zero). Hence, increasing $\theta$ means increasing the coherence in both the states and effects.
	
	Ref.~\cite{DAVID1} discusses the consequences of noncontextuality for this scenario. Ref.~\cite{DAVID1} shows that quantum theory allows for a higher probability of success at distinguishing $\ket{\psi}$ from $\ket{\phi}$ than is possible in any noncontextual theory.
	Hence, there is a quantum advantage for this task coming from contextuality. Ref.~\cite{DAVID1}  also analytically estimates how much depolarizing noise, $r_{\text{min}}$, must be added to the quantum model until this quantum advantage disappears, i.e., until one's quantum measurements perform no better than noncontextual measurements could.  
	Starting from the expression for depolarized effects
	\begin{equation}
		E\mapsto\mathcal{D}_r^{\text{depol}}(E):=(1-r)E+\frac{r}{2} \mathbb{1},
	\end{equation}
	imposing  the existence of a noncontextual model for the quantum scenario implies that
	\begin{equation}
		r^\text{depol}_\text{min} =1-\frac{1}{\sin^2\theta+\cos\theta}.
	\end{equation}
	This was first computed in Ref.~\cite{DAVID1}, although with a minor error that we correct in our proof in Appendix \ref{app1}.
	
	\begin{figure}[t]
		\centering
		\begin{tabular}{cc}
			\adjustbox{scale=0.75,valign=m}{\begin{tikzpicture}
			\begin{pgfonlayer}{background}
		\fill[gray!20] (9.175, 3.5) to[bend right=45] (7.675, 2) to[bend right=45] (9.175, 0.5) to[bend right=45] (10.675, 2) to[bend right=45] (9.175, 3.5);
		\fill [white] (9.175, 2.875) to[bend right=45] (8.3, 2) to[bend right=45] (9.175, 1.125) to[bend right=45] (10.05, 2) to[bend right=45]  cycle;
	\end{pgfonlayer}
	\begin{pgfonlayer}{nodelayer}
		\node [style=hadamard] (1) at (9.175, 2.875) {};
		\node [style=hadamard] (2) at (9.175, 1.125) {};
		\node [style=none] (3) at (9.175, 3.925) {$E_{g_\psi}$};
		\node [style=none] (4) at (9.175, 0) {$E_{g_\phi}$};
		\node [style=none] (5) at (9.5, 2.25) {};
		\node [style=none] (7) at (9.175, 2.875) {};
		\node [style=none] (8) at (8.225, 2) {};
		\node [style=none] (9) at (9.175, 1.125) {};
		\node [style=none] (10) at (10.05, 2) {};
		\node [style=none] (13) at (10.875, 3.2) {$\ket{\psi}\bra{\psi}$};
		\node [style=hadamard] (11) at (9.9, 2.55) {};
		\node [style=hadamard] (12) at (8.45, 1.45) {};
		\node [style=hadamard] (14) at (8.45, 2.55) {};
		\node [style=hadamard] (15) at (9.9, 1.45) {};
		\node [style=none] (16) at (7.525, 3.2) {$\ket{\bar{\phi}}\bra{\bar{\phi}}$};
		\node [style=none] (17) at (7.5, 0.75) {$\ket{\bar{\psi}}\bra{\bar{\psi}}$};
		\node [style=none] (18) at (10.8, 0.75) {$\ket{\phi}\bra{\phi}$};
		\node [style=none] (26) at (9.175, 3.5) {};
		\node [style=none] (27) at (7.675, 2) {};
		\node [style=none] (28) at (9.175, 0.5) {};
		\node [style=none] (29) at (10.675, 2) {};
		\node [style=none] (30) at (7.75, 2) {};
		\node [style=none] (31) at (8.025, 1.725) {$r$};
		\node [style=none] (32) at (8.3, 2) {};
		\node [style=none] (33) at (9.175, 2.25) {};
		\node [style=none] (34) at (9.352, 2.153) {};
		\node [style=none] (35) at (9.35, 2.4) {$\theta$};
		\node [style=none] (36) at (10.375, 2.91) {};
		\node [style=none] (37) at (7.945, 1.07) {};
		\node [style=none] (38) at (10.375, 1.09) {};
		\node [style=none] (39) at (7.945, 2.93) {};
	\end{pgfonlayer}
	\begin{pgfonlayer}{edgelayer}
		\draw (1) to (2);
		\draw (12) to (11);
		\draw (15) to (14);
		\draw [bend left=45] (27.center) to (26.center);
		\draw [bend left=45] (26.center) to (29.center);
		\draw [bend left=45] (29.center) to (28.center);
		\draw [bend left=45] (28.center) to (27.center);
		\draw [dashed] (26.center) to (28.center);
		\draw [|-|] (30.center) to (8.center);
		\draw [bend left=45] (32.center) to (7.center);
		\draw [bend right=45] (10.center) to (7.center);
		\draw [bend right=45] (32.center) to (9.center);
		\draw [bend left=45] (10.center) to (9.center);
		\draw [bend left] (33.center) to (34.center);
		\draw [dashed] (36) to (37);
		\draw [dashed] (38) to (39);
	\end{pgfonlayer}
\end{tikzpicture}} & \adjustbox{valign=m}{\includegraphics[scale=0.4]{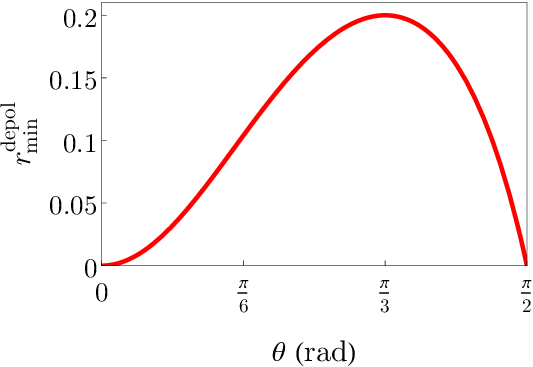}}
		\end{tabular}
		\caption{Left: representation of the action of {\bf depolarizing noise} on the scenario. Right: analytical plot of contextual robustness to depolarization as a function of the angle between the prepared states and the $Z$ axis.}
		\label{Fig2}
	\end{figure}
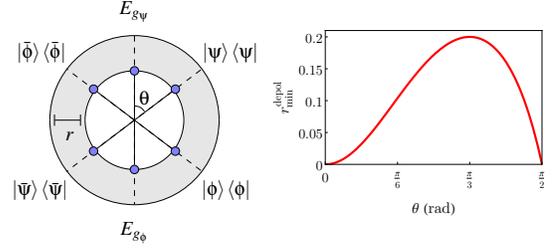
	
	This is plotted in Figure~\ref{Fig2}, from which we can see that even small values of $\theta$ allow for a proof of contextuality. Indeed, the robustness to depolarization is null only for $\theta=0$ and $\theta=\frac{\pi}{2}$, circumstances in which $\ket{\psi}$ is equal to either $\ket{\phi}$ or $\ket{\bar{\phi}}$. In this case, the cardinality of the sets of states/effects decreases, and simplex-embeddability becomes possible. Since coherence vanishes as $\theta$ goes to zero, this establishes the following result (which was not previously recognized, although it requires only the results of Ref.~\cite{DAVID1} reiterated just above):
	\begin{result}
		{\em There are proofs of the failure of noncontextuality that can be achieved in a prepare-and-measure scenario with vanishing (but nonzero) coherence among both the states and the effects.}
	\end{result}
	
	This particular scenario is not very robust to depolarization noise. Even at its peak \footnote{Interestingly, this peak occurs at $\theta=\frac{\pi}{3}$, where one finds the same set of effects considered in Ref.~\cite{SPEKKENS} to obtain the first proof of the failure of measurement noncontextuality for POVMs in a two-dimensional system.}, for $\theta=\frac{\pi}{3}$, the robustness to depolarization is   $0.2$. Moreover, the robustness goes smoothly to zero as the coherence goes to 0, and consequently is very small for small coherence. As we will see, this same scenario is considerably more robust to dephasing, and a closely related scenario is in fact maximally robust to dephasing.
	
	\section{Proof of contextuality maximally robust to dephasing}
	Next, we explore how contextuality behaves in this scenario under dephasing (rather than depolarizing) noise.
	In this case, the noisy projectors are given by
	\begin{equation}  E\mapsto\mathcal{D}_r^{\text{deph}}(E)\\:=(1-r)E+r\sum_{i\in\{0,1\}}\braket{i|E|i}|i\rangle\langle i|,
	\end{equation}
	where $\{\ket{i}\}_{i\in\{0,1\}}$ is the $Z$ basis. Imposing the existence of a noncontextual model for the scenario implies that
	\begin{equation}
		r^\text{deph}_\text{min} =1-\frac{1-\cos\theta}{\sin^2\theta}
	\end{equation}
	as proven in Appendix \ref{app1} and plotted in Figure \ref{Fig3}.

	\begin{figure}[t]
		\centering
		\begin{tabular}{cc}
			\adjustbox{scale=0.75,valign=m}{\begin{tikzpicture}

	\begin{pgfonlayer}{nodelayer}
		\node [style=hadamard] (0) at (3.925, 3.5) {};
		\node [style=hadamard] (1) at (3.925, 0.5) {};
		\node [style=none] (2) at (3.925, 3.925) {$E_{g_\psi}$};
		\node [style=none] (3) at (3.925, 0) {$E_{g_\phi}$};
		\node [style=none] (7) at (3.925, 3.5) {};
		\node [style=none] (8) at (2.425, 2) {};
		\node [style=none] (9) at (3.925, 0.5) {};
		\node [style=none] (10) at (5.425, 2) {};
		\node [style=hadamard] (11) at (4.45, 2.975) {};
		\node [style=hadamard] (12) at (3.4, 0.955) {};
		\node [style=none] (13) at (5.625, 3.575) {$\ket{\psi}\bra{\psi}$};
		\node [style=hadamard] (14) at (3.4, 2.975) {};
		\node [style=hadamard] (15) at (4.45, 0.955) {};
		\node [style=none] (16) at (2.325, 3.575) {$\ket{\bar{\phi}}\bra{\bar{\phi}}$};
		\node [style=none] (17) at (2.25, 0.375) {$\ket{\bar{\psi}}\bra{\bar{\psi}}$};
		\node [style=none] (18) at (5.55, 0.375) {$\ket{\phi}\bra{\phi}$};
		\node [style=none] (19) at (2.525, 1.9625) {};
		\node [style=none] (20) at (3.175, 1.9625) {};
		\node [style=none] (21) at (2.84, 1.762) {$r$};
		\node [style=none] (22) at (5.025, 2.975) {};
		\node [style=none] (23) at (2.835, 0.955) {};
		\node [style=none] (24) at (5.025, 0.955) {};
		\node [style=none] (25) at (2.835, 2.975) {};
		\node [style=none] (26) at (3.925, 2.4625) {};
		\node [style=none] (27) at (4.265,2.2685) {};
		\node [style=none] (28) at (4.05,2.7) {$\theta$};
	\end{pgfonlayer}
	\begin{pgfonlayer}{edgelayer}
		\draw (1) to (0);
		\draw [bend left=45] (8.center) to (7.center);
		\draw [bend left=45] (7.center) to (10.center);
		\draw [bend left=45] (10.center) to (9.center);
		\draw [bend left=45] (9.center) to (8.center);
		\draw (12) to (11);
		\draw (15) to (14);
		\draw [bend left=30] (0.center) to (11.center);
		\draw [bend left=15] (11.center) to (15.center);
		\draw [bend left=30] (15.center) to (1.center);
		\draw [bend left=30] (1.center) to (12.center);
		\draw [bend left=15] (12.center) to (14.center);
		\draw [bend left=30] (14.center) to (0.center);
		\draw [{|-|}] (20.center) to (19.center);
		\draw [dashed] (22.center) to (23.center);
		\draw [dashed] (24.center) to (25.center);
		\draw [bend left] (26.center) to (27.center);
	\end{pgfonlayer}

	\begin{pgfonlayer}{background}
	\fill[gray!20] (7) to[bend right=45] (8) to[bend right=45] (9) to[bend right=45] (10) to[bend right=45] (7);
	\fill [white] (0) to[bend left=30] (11) to[bend left=15] (15) to[bend left=30] (1) to[bend left=30]  (12) to[bend left=15] (14) to[bend left=30] (0);
\end{pgfonlayer}
\end{tikzpicture}} &
			\adjustbox{valign=m}{\includegraphics[scale=0.4]{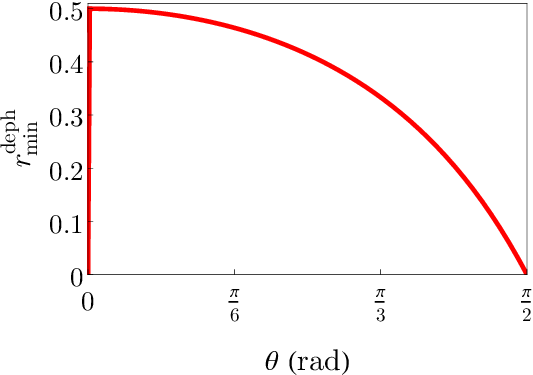}}
		\end{tabular}
		\caption{Left: representation of the action of {\bf dephasing noise} on the plane in which measurements live. Right: analytical plot of contextual robustness to dephasing as a function of the angle between the prepared states and the $Z$ axis.}
		\label{Fig3}
	\end{figure}
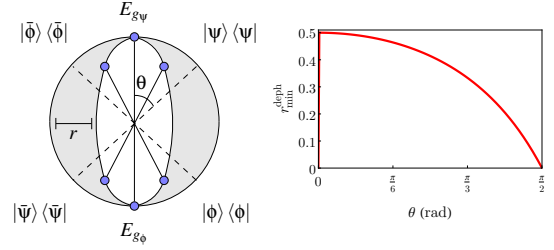
	
	Figure \ref{Fig3} shows that the amount of contextuality {\em decreases} monotonically as coherence increases (at least according to these measures of contextuality and coherence). While this might at first seem counterintuitive, one can understand it by noting that even under large dephasing noise, states and effects with little coherence to begin with are barely affected; that is, the dephasing channel is close to identity on such states and effects. However, this intuition only goes so far, as we will give an example below where one achieves twice the robustness to dephasing by including effects that have maximal coherence in the dephased basis.
	Just as was the case for depolarizing noise, the robustness to dephasing drops to zero when $\ket{\psi}=\ket{\phi}$, which implies a discontinuity in the plot at $\theta=0$, where  $r^\text{deph}_{\text{min}}$ falls from $0.5$ to $0$.
	
	Moreover, note that the maximum robustness relative to dephasing (0.5) is much higher than the maximal robustness to depolarization (0.2).
	A natural question is whether one can find scenarios where the contextual robustness to dephasing approaches its logical maximum, or whether (as for the contextual robustness to depolarization) this quantity is always bounded from above. 
	In the following, we identify such a scenario by carrying out a numerical exploration using a modified version of the linear program from Ref.~\cite{CODE}.
	
	Consider modifying the above scenario by rotating the measurement $M_g$ from the $Z$ to the $X$ axis, as in Figure~\ref{Fig4}, thus maximizing the coherence in the effects associated with this measurement. 
	In the case of depolarizing noise, this scenario is equivalent by symmetry to the original scenario under relabeling. Here, however, we are interested in the case of dephasing noise in the $Z$ basis. (Note that this scenario is equivalent to the original scenario under relabeling and considering the noise to be in the $X$ basis rather than the $Z$ basis.)
	The robustness to dephasing for this scenario as a function of $\theta$ is plotted in Figure~\ref{Fig4}. The most striking feature of this plot is that the robustness approaches $1$ as $\theta\to0$, so that the scenario achieves the maximum logically possible robustness to dephasing. This is in stark contrast with the original scenario (where the maximal dephasing robustness was $0.5$). 
	
	Moreover, notice that if we start with the scenario  described in Fig.  \ref{Fig4} and then dephase by some $r$ such that $1-r$ is vanishingly small, then we can view the dephased scenario as a new scenario that has only a vanishing amount of coherence in the measurements, but which is still contextual and indeed is itself robust to arbitrary amounts of dephasing noise. (This follows from the fact that dephasing it by some $r'$ is the same as dephasing the original scenario by $1-(1-r')(1-r)>0$.) Thus, we have established the following:
	\begin{result}
		{\em There are proofs of the failure of noncontextuality that can be achieved in a prepare-and-measure scenario in the presence of arbitrarily large dephasing noise. One may moreover find some scenarios of this sort where the states and effects have vanishing (but nonzero) coherence.}
	\end{result}
	
	\begin{figure}[t!]
		\centering
		\begin{tabular}{cc}
			\adjustbox{scale=0.7,valign=m}{\begin{tikzpicture}
	
	\begin{pgfonlayer}{nodelayer}
		\node [style=hadamard] (0) at (4.6, 1.9625) {};
		\node [style=hadamard] (1) at (3.25, 1.9625) {};
		\node [style=none] (2) at (5.9, 1.9625) {$E_{g_\psi}$};
		\node [style=none] (3) at (2, 1.9625) {$E_{g_\phi}$};
		\node [style=none] (7) at (3.925, 3.5) {};
		\node [style=none] (8) at (2.425, 2) {};
		\node [style=none] (9) at (3.925, 0.5) {};
		\node [style=none] (10) at (5.425, 2) {};
		\node [style=hadamard] (11) at (4.45, 2.975) {};
		\node [style=hadamard] (12) at (3.4, 0.955) {};
		\node [style=none] (13) at (5.625, 3.575) {$\ket{\psi}\bra{\psi}$};
		\node [style=hadamard] (14) at (3.4, 2.975) {};
		\node [style=hadamard] (15) at (4.45, 0.955) {};
		\node [style=none] (16) at (2.325, 3.575) {$\ket{\bar{\phi}}\bra{\bar{\phi}}$};
		\node [style=none] (17) at (2.25, 0.375) {$\ket{\bar{\psi}}\bra{\bar{\psi}}$};
		\node [style=none] (18) at (5.55, 0.375) {$\ket{\phi}\bra{\phi}$};
		\node [style=none] (19) at (2.525, 1.9625) {};
		\node [style=none] (20) at (3.15, 1.9625) {};
		\node [style=none] (21) at (2.84, 1.762) {$r$};
		\node [style=none] (22) at (5.025, 2.975) {};
		\node [style=none] (23) at (2.835, 0.955) {};
		\node [style=none] (24) at (5.025, 0.955) {};
		\node [style=none] (25) at (2.835, 2.975) {};
		\node [style=none] (26) at (3.925, 2.4625) {};
		\node [style=none] (27) at (4.265,2.2685) {};
		\node [style=none] (28) at (4.05,2.7) {$\theta$};
	\end{pgfonlayer}
	\begin{pgfonlayer}{edgelayer}
	\draw (1) to (0);
	\draw [dashed] (7) to (9);
	\draw [bend left=45] (8.center) to (7.center);
	\draw [bend left=45] (7.center) to (10.center);
	\draw [bend left=45] (10.center) to (9.center);
	\draw [bend left=45] (9.center) to (8.center);
	\draw (12) to (11);
	\draw (15) to (14);
	\draw [bend left=30] (7.center) to (11.center);
	\draw [bend left=15] (11.center) to (15.center);
	\draw [bend left=30] (15.center) to (9.center);
	\draw [bend left=30] (9.center) to (12.center);
	\draw [bend left=15] (12.center) to (14.center);
	\draw [bend left=30] (14.center) to (7.center);
	\draw [{|-|}] (20.center) to (19.center);
	\draw [dashed] (22.center) to (23.center);
	\draw [dashed] (24.center) to (25.center);
	\draw [bend left] (26.center) to (27.center);
\end{pgfonlayer}

	\begin{pgfonlayer}{background}
		\fill[gray!20] (7) to[bend right=45] (8) to[bend right=45] (9) to[bend right=45] (10) to[bend right=45] (7);
		\fill [white] (7) to[bend left=30] (11) to[bend left=15] (15) to[bend left=30] (9) to[bend left=30]  (12) to[bend left=15] (14) to[bend left=30] (7);
	\end{pgfonlayer}
\end{tikzpicture}} &
			\adjustbox{valign=m}{\includegraphics[scale=0.4]{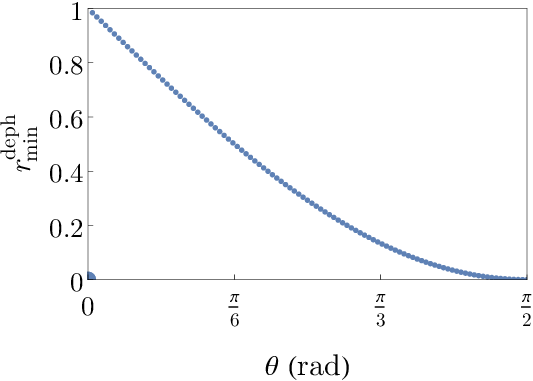}}
		\end{tabular}
		\caption{Left: representation of the action of {\bf dephasing noise} on the plane in which measurements of the rotated scenario live. Right: numeric estimation of contextual robustness to dephasing as a function of the angle between the prepared states and the $Z$ axis, with the extra measurement now lying along the $X$ axis.}
		\label{Fig4}
	\end{figure}
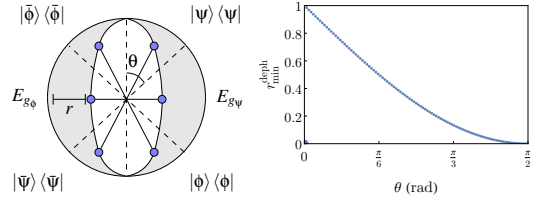
	
	Based on the intuitive arguments above, one may have expected the original scenario (Fig.~\ref{Fig3}) to be more robust than this rotated one (Fig.~\ref{Fig4}), since the highly coherent effects in the latter case are strongly affected by dephasing noise. However, this effect is clearly more than compensated for by the fact that the effects $E_{g_\phi}$ and $E_{g_\psi}$, in this case, are further from the other states and effects in the scenario, thus making it more difficult to find a simplex embedding of the GPT.
	
	In Appendix \ref{app4} we present a numerical plot scanning between the case in which measurement $M_g$ is aligned with the $Z$ axis and when it is aligned with the $X$ axis, showing that the latter is indeed the only scenario in this family for which maximal robustness is achieved.
	
	A final natural question is how contextuality is affected in the presence of both dephasing and depolarizing noise. In Appendix \ref{app4}, we study contextual robustness to dephasing in a scenario where a small amount of depolarizing noise is added to the measurement $M_g$, and show that even when $M_g$ is partially depolarized, proofs of contextuality can be obtained in the dephased scenario as long as the depolarizing noise on $M_g$ does not surpass a bound given by the amount of coherence available from the states.
	
	\section{Discussion}\label{sec3} 
	
	We have exhibited scenarios in which any nonzero amount of coherence is enough to prove the failure of the assumption of noncontextuality, and introduced an example in which a proof of contextuality can be robust to any amount of dephasing noise other than complete dephasing. This work  showcases the versatility of the linear program from Ref.~\cite{CODE} and invites further research on how robust specific proofs of contextuality are under different types of noise.
	
	Another recent work explored the connection between coherence and contextuality, using event graphs~\cite{RAFA}. Violations of graph inequalities that witness both basis-independent coherence and contextuality~\cite{OVERLAP, SET} were derived in Ref.~\cite{RAFA} and applied to a similar scenario \cite{RAFA2}, in this case with $6$ preparations rather than $4$. These inequalities are not violated for the whole interval $\theta\in\left(0,\frac{\pi}{2}\right)$, in contrast to the scenario studied herein.
	To get some intuition on why, notice that the existence of basis-independent coherence is not sufficient to guarantee the failure of noncontextuality. To see this, recall that contextuality always goes to zero by partial depolarization, and yet the only depolarizing process that destroys all basis-independent coherence is the totally depolarizing process. 
	
	\section*{Acknowledgements} 
	We thank Elie Wolfe for the ready support with the linear program functioning and Rafael Wagner for the fruitful discussion. D.S. and V.R. were supported by the Foundation for Polish Science (IRAP project, ICTQT, Contract No. MAB/2018/5, cofinanced by EU within Smart Growth Operational Programme). J.H.S. was supported by the National Science Centre, Poland (Opus project, Categorical Foundations of the Non-Classicality of Nature, Project No. 2021/41/B/ST2/03149). A.B.S. acknowledges support by the Digital Horizon Europe project FoQaCiA, Foundations of quantum computational advantage, GA No. 101070558, funded by the European Union, NSERC (Canada), and UKRI (U.K.). Some figures were prepared using TikZit.
	
	\bigskip
	\onecolumngrid
	\appendix
	\vspace{-20pt}\section{Analytical derivation of the robustness to depolarization and dephasing}\label{app1}

	\subsection{Robustness to depolarization} 
	The prepare-and-measure scenario of interest for our purposes consists of $4$ preparations denoted $P \in\{P_\psi, P_{\bar{\psi}}, P_\phi, P_{\bar{\phi}}\}=:\mathcal{P}$ and $3$ binary-outcome measurements $k|M\in \{0,1\}\times\{M_\phi,M_\psi,M_g\}=:K\times\mathcal{M}$. Without any further restrictions, there would be $12$ free parameters in the data table $\{p(k|M,P)\}_{P\in\mathcal{P},[k|M]\in K\times\mathcal{M}}$. However, Ref.~\cite{DAVID1} imposes further constraints on these preparations and measurement outcomes, based on the symmetries of $\mathcal{P}$ and $\mathcal{M}$ and on the inferential equivalence 
	\begin{equation}\label{equiv}
		\frac12 P_\psi+\frac12 P_{\bar{\psi}}=\frac12 P_\phi+\frac12 P_{\bar{\phi}}.
	\end{equation}
	As a result, any data table satisfying these will necessarily have only three free parameters. In particular, Ref.~\cite{DAVID1} takes the three parameters to be denoted by $s$, $c$, and $\epsilon$, which are related to the statistics of the operational theory as represented in Table \ref{Table1}.
	\begin{table}[h]
		\centering
		\begin{tabular}{c|cccc}\hline
			$p(k|M,P)$ & $P_\phi$ & $P_\psi$ & $P_{\bar{\phi}}$ & $P_{\bar{\psi}}$\\\hline
			$0|M_\phi$ & $1-\epsilon$ & $c$ & $\epsilon$ & $1-c$\\
			$0|M_\psi$ & $c$ & $1-\epsilon$ & $1-c$ & $\epsilon$\\
			$0|M_g$ & $s$ & $1-s$ & $1-s$ & $s$\\\hline
		\end{tabular}
		\caption{Possible statistics for an operational theory with 4 preparations and 3 binary measurements, satisfying the equivalence $\frac12P_\psi+\frac12P_{\bar{\psi}}=\frac12P_\phi+\frac12P_{\bar{\phi}}$. Note that we only show the $k=0$ outcome statistic for each measurement as the $k=1$ are then uniquely determined by normalization.}
		\label{Table1}
	\end{table}
	
	The demand that there is a noncontextual ontological model for this operational theory forces a relation upon these three parameters. The derivation of this relation is well explained in Ref.~\cite{DAVID1} and culminates in solving a rather extensive linear system, and is given by the inequality
	\begin{equation}\label{ineq}
		s\leq 1-\frac{c-\epsilon}{2}.
	\end{equation}
	If a data table parameterized as in Table \ref{Table1} satisfies this condition, then it can be explained by a generalized noncontextual ontological model, otherwise it cannot.
	
	In our quantum experiment, described in Fig.~1 of the main text, we find that the  inferential  equivalence \ref{equiv} is satisfied and that we can write the three parameters in the data table as functions of $\theta$ and $r$.
	To do so, first note that the depolarizing noise acts on the effects of the scenario such that
	\begin{equation}
		E\mapsto\mathcal{D}_r^{\text{depol}}(E):=(1-r)E+\frac{r}{2} \mathbb{1},
	\end{equation}
	where $\mathbb{1}$ is the $2\times2$ identity matrix. This allows us to compute the statistics of our experiment for depolarized measurements in Table \ref{Table2}. 
	\begin{table}[h]
		\centering
		\begin{tabular}{c|cccc}\hline
			$\text{Born rule}$ & $\ket{\phi}\!\!\bra{\phi}$ & $\ket{\psi}\!\!\bra{\psi}$ & $\ket{\bar{\phi}}\!\!\bra{\bar{\phi}}$ & $\ket{\bar{\psi}}\!\!\bra{\bar{\psi}}$\\\hline
			$\mathcal{D}_r^{\text{depol}}(\ket{\phi}\!\!\bra{\phi})$ & $1-\frac{r}{2}$ & $(1-r)\sin^2\theta+\frac{r}{2}$ & $\frac{r}{2}$ & $(1-r)\cos^2\theta+\frac{r}{2}$\\
			$\mathcal{D}_r^{\text{depol}}(\ket{\psi}\!\!\bra{\psi})$ & $(1-r)\sin^2\theta+\frac{r}{2}$ & $1-\frac{r}{2}$ & $(1-r)\cos^2\theta+\frac{r}{2}$ & $\frac{r}{2}$\\
			$\mathcal{D}_r^{\text{depol}}(E_{g_\phi})$ & $\frac{1-r}{2}(1+\cos\theta)+\frac{r}{2}$ & $\frac{1-r}{2}(1-\cos\theta)+\frac{r}{2}$ & $\frac{1-r}{2}(1-\cos\theta)+\frac{r}{2}$ & $\frac{1-r}{2}(1+\cos\theta)+\frac{r}{2}$\\\hline
		\end{tabular}
		\caption{Statistics of the prepare-and-measure scenario with projectors depolarized by a parameter $r$.}
		\label{Table2}
	\end{table}

	By comparing the two tables we can write $s$, $c$, and $\epsilon$ as functions of $\theta$ and $r$:
	\begin{eqnarray}
		s&=&\frac12+\frac{1-r}{2}\cos\theta;\\
		c&=&(1-r)\sin^2\theta+\frac{r}{2};\\
		\epsilon&=&\frac{r}{2}.
	\end{eqnarray}
	From this, we can compute the minimal amount of depolarizing noise $r^\text{depol}_{\text{min}}$, which is required so that the quantum experiment admits of a noncontextual ontological model. That is, the value of $r$ such that the inequality \ref{ineq} is satisfied tightly. The value obtained for $r^\text{depol}_{\text{min}}$ will depend on $\theta$ and so we obtain the equation
	\begin{equation}
		r^\text{depol}_\text{min}=1-\frac{1}{\sin^2\theta+\cos\theta},
	\end{equation}
	which is the equation that was shown and discussed in the main text. Notice that in Ref.~\cite{DAVID1} the depolarizing noise acts both on states and effects, so the statistics in Table \ref{Table2} in our case provide different values for the parameters $s$, $c$, and $\epsilon$ than in that work. However, the plots of $r^\text{depol}_\text{min}$ coincide here and in Ref.~\cite{DAVID1} (up to a reparameterization of $\theta$) due to a miscalculation in the latter, which re-scaled $r^\text{depol}_\text{min}$ to the case of depolarizing noise acting only on the effects. Because the two cases -- depolarizing noise acting only on effects or acting on both effects and measurements -- are equivalent up to this reparameterization, there is no impact on any of the analyzes in Ref.~\cite{DAVID1}.
	
	\subsection{Robustness to dephasing} 
	We can now repeat the same analysis from the previous section, but in this case considering the scenario described in Fig.~3 of the main text, where $r$ parameterizes the amount of dephasing rather than depolarizing noise. Recall that in this scenario the dephased effects are given by
	\begin{equation}
		E\mapsto\mathcal{D}_r^{\text{deph}}(E)\\:=(1-r)E+r\sum_{i\in\{0,1\}}\braket{i|E|i}|i\rangle\langle i|,
	\end{equation}
	where $\{\ket{i}\}_{i\in\{0,1\}}$ is the $Z$ basis. The statistics for this new scenario can then be computed and are given in Table \ref{Table3}.
	
	\begin{table}[h]
		\centering
		\begin{tabular}{c|cccc}\hline
			$\text{Born rule}$ & $\ket{\phi}\!\!\bra{\phi}$ & $\ket{\psi}\!\!\bra{\psi}$ & $\ket{\bar{\phi}}\!\!\bra{\bar{\phi}}$ & $\ket{\bar{\psi}}\!\!\bra{\bar{\psi}}$\\\hline
			$\mathcal{D}_r^{\text{deph}}(\ket{\phi}\!\!\bra{\phi})$ & $1-\frac{r}{2}$ & $(1-r)\sin^2\theta+\frac{r}{2}$ & $\frac{r}{2}$ & $(1-r)\cos^2\theta+\frac{r}{2}$\\
			$\mathcal{D}_r^{\text{deph}}(\ket{\psi}\!\!\bra{\psi})$ & $(1-r)\sin^2\theta+\frac{r}{2}$ & $1-\frac{r}{2}$ & $(1-r)\cos^2\theta+\frac{r}{2}$ & $\frac{r}{2}$\\
			$\mathcal{D}_r^{\text{deph}}(E_{g_\phi})$ & $\frac12(1+\cos\theta)$ & $\frac12(1-\cos\theta)$ & $\frac12(1-\cos\theta)$ & $\frac12(1+\cos\theta)$\\\hline
		\end{tabular}
		\caption{Statistics of the prepare-and-measure scenario, now with projectors dephased by a parameter $r$. Notice that the first two rows are the same as in the depolarized case (Table \ref{Table2}). However, the entries in the last row do not depend on $r$ since the dephasing noise does not change the effects aligned with the $Z$ axis. Hence, the third row here is the same as the third row of Table \ref{Table2} with $r$ set to zero.}
		\label{Table3}
	\end{table}
	
	Like in the previous section, we can then compare this to Table \ref{Table1} in order to write the parameters $s$, $c$, and $\epsilon$ as functions of $r$ and $\theta$:
	\begin{eqnarray}
		s&=&\frac12(1+\cos\theta);\\
		c&=&(1-r)\sin^2\theta+\frac{r}{2};\\
		\epsilon&=&\frac{r}{2}.
	\end{eqnarray}
	Note that $c$ and $\epsilon$ are the same as in the depolarizing case, but that $s$ is now independent of $r$.
	
	Finally, we can compute the maximal robustness to dephasing noise by demanding that inequality \ref{ineq} is saturated, that is, imposing the minimum dephasing noise $r^\text{deph}_\text{min}$ necessary for the existence of a noncontextual ontological model. This leads to the following equation for $r^\text{deph}_{\text{min}}$,
	\begin{equation}
		r^\text{deph}_\text{min}=1-\frac{1-\cos\theta}{\sin^2\theta},
	\end{equation}
	which is exactly what we gave and discussed in  the main text.

	\section{Failures of noncontextuality cannot be achieved without set coherence}\label{app2}
	
	In this section, we  give formal proof that one cannot prove the failure of the noncontextuality in scenarios where all the states (or all the effects) have no set coherence---that is, are simultaneously diagonalizable~\cite{SET}. This means that computing robustness with respect to dephasing noise is a sensible measure of the failure of the existence of a noncontextual ontological model, as under sufficient dephasing noise \emph{any} scenario will admit of a noncontextual ontological model. This result is well known in the community but we are not aware of an explicit proof so we include it here for convenience. \blk
	
	\begin{proposition}
		Incoherent quantum states or measurements cannot prove the failure of noncontextuality.
	\end{proposition}
	\begin{proof}
		We here give the proof for the case of incoherent states, the case of incoherent measurements following straightforwardly. Consider a quantum system $\mathcal{H}$, a set of quantum states $\Omega:=\{\rho_P\}_{P\in\mathcal{P}}$ for this system, and quantum effects $\mathcal{E}:=\{E_{k|M}\}_{[k|M]\in K\times\mathcal{M}}$ acting on the system, such that $\sum_{k\in K}E_{k|M}=\mathbb{1}$, $\forall M\in\mathcal{M}$. Let $\{\ket{i}\}_{i\in I}$ be the basis in which all $\rho_P$ are diagonalised, $I=\{0,1,...,\text{dim}\mathcal{H}\}$. We define a linear map $\mu:\Omega\to D[I]::\rho_P\mapsto \mu_P$ where $D[I]\subset \mathbb{R}^I$ is the space of probability distributions over the index set $I$, where the $\mu_P$ are defined pointwise by
		\begin{equation}
			\mu_P(i):=\Tr\{\ket{i}\!\!\bra{i}\rho_P\}=\braket{i|\rho_P|i},\quad\forall P\in\mathcal{P},\ \forall i\in I.
		\end{equation}
		These are indeed valid probability distributions as it is easy to show that $\sum_{i\in I}\mu_P(i)=1$, $\forall P\in\mathcal{P}$.  We also define a linear map $\xi:\mathcal{E}\to R[I]::E_{k|M} \mapsto \xi_{k|M}$ where $R[I]\subset \mathbb{R}^I$ is the space of response functions over the index set $I$, where the $\xi_{k|M}$ are defined pointwise by \blk
		\begin{equation}
			\xi_{k|M}(i):=Tr\{E_{k|M}\ket{i}\!\!\bra{i}\}=\braket{i|E_{k|M}|i},\quad\forall[k|M]\in K\times\mathcal{M},\ \forall i \in I,
		\end{equation}
		such that $\xi_{k|M}(i) \in [0,1]$ $\forall i\in I$ and $\sum_{k\in K}\xi_{k|M}(i)=1$ $\forall M\in\mathcal{M}, i\in I$, hence these constitute valid response functions.
		
		Notice now that the quantum statistics in the experiment are reproduced by these probability distributions and response functions, since
		\begin{eqnarray}
			\Tr\{E_{k|M}\rho_P\} &=& \sum_{i\in I}\braket{i|E_{k|M}\rho_P|i}\\
			&=&\sum_{i,j\in I}\braket{i|E_{k|M}|j}\braket{j|\rho_P|i}\\
			&=& \sum_{i,j\in I}\braket{i|E_{k|M}|j}\braket{i|\rho_P|i}\delta_{ij}\\
			&=&\sum_{i\in I}\braket{i|E_{k|M}|i}\braket{i|\rho_P|i}\\
			&=&\sum_{i\in I}\xi_{k|M}(i)\mu_P(i),
		\end{eqnarray}
		where for the third equality we used the fact that for all $P\in\mathcal{P}$, $\rho_P$ is diagonal in the basis $\{\ket{i}\}_{i\in I}$. If we instead were working with incoherent measurements in this step, we would have instead used that $\braket{i|E_{k|M}|j}=\delta_{ij}\braket{i|E_{k|M}|i}$ in order to obtain the same result.
		
		Finally, notice that whenever $\rho_P=\rho_{P'}$ (resp. $E_{k|M}=E_{k'|M'}$), it will be the case that $\mu_P(i)=\mu_{P'}(i)$ (resp. $\xi_{k|M}(i)=\xi_{k'|M'}(i)$), $\forall i\in I$, therefore constituting a noncontextual ontological model for the statistics of the scenario. \blk
	\end{proof}
	
	\section{Robustness to dephasing with the linear program}\label{app3} 
	\subsection{Formal definitions}
	
	We begin this section by giving a formal definition of a GPT description of a given operational prepare-measure scenario  \cite{SIMPLEX,GPT2}:
	\begin{definition} A \emph{generalized probabilistic theory} associated with the operational scenario $(\mathcal{P},\mathcal{M},K,p)$ is a tuple $(V,\braket{\cdot,\cdot},\Omega,\mathcal{E})$ such that
		\begin{compactitem}
			\item $(V,\braket{\cdot,\cdot})$ is a finite-dimensional, real vector space equipped with an inner product;
			\item $\Omega\subset V$ is a compact, convex set such that $V=\mathsf{LinSpan}[\Omega]$ and $0 \not\in \mathsf{AffSpan}[\Omega]$, and where any element $s\in\Omega$, called a \emph{state}, is associated with an inferential equivalence class of preparations, $\widetilde{P}\in\mathcal{P}/\simeq$; 
			\item $\mathcal{E}$ is a subset of the dual $\Omega^*$, such that both the origin $0$ and the unit $u$ (i.e., the unique vector satisfying $\braket{u,s}=1$ for all $s\in \Omega$) in $\Omega^*$ are in $\mathcal{E}$, where any element $\varepsilon\in\mathcal{E}$ is called an \emph{effect} and is associated with an inferential equivalence class of measurement outcomes, $\widetilde{[k|M]}\in\mathcal{M}/\simeq$;
			\item For all $[k|M]\in\mathcal{M}$ and $P\in\mathcal{P}$, there is a respective $\varepsilon\in\mathcal{E}$ and $s\in\Omega$ such that
			\begin{equation}
				p(k|M,P)=\braket{\varepsilon,s}.
			\end{equation}
		\end{compactitem}
	\end{definition}
	A GPT is therefore a geometrical description of the  operational scenario in which we have quotiented the sets of preparations and measurement outcomes by the inferential equivalence relation, since variations within the equivalence classes (i.e., the context of the procedure) are irrelevant for making predictions. In particular, this means that states and effects within the GPT satisfy the principle of \emph{tomography}, i.e.,
	\begin{equation}
		\braket{\varepsilon,s_1}=\braket{\varepsilon,s_2},\,\forall \varepsilon\in\mathcal{E}\iff s_1=s_2;\quad\braket{\varepsilon_1,s}=\braket{\varepsilon_2,s},\,\forall s\in\Omega\iff \varepsilon_1=\varepsilon_2.
	\end{equation}
	
	The notion of nonclassicality employed in this work is the existence of a noncontextual ontological model of the operational scenario, which was shown in Ref.~\cite{SIMPLEX} to be equivalent to the simplex-embeddability of the associated GPT. The latter is defined as follows:
	\begin{definition} A GPT $(V,\braket{\cdot,\cdot},\Omega,\mathcal{E})$ is \emph{simplex-embeddable}  if and only if \\
		(i) there exists $n\in \mathbb{N}$ defining
		\begin{compactitem}
			\item the real vector space $\mathbb{R}^n$ with Euclidean inner product $\_\cdot\_ $, and
			\item the unit simplex $\Delta_n\in \mathbb{R}^n$ and its dual, the unit hypercube $\Delta^*_n$, 
		\end{compactitem}
		(ii) there exists a pair of linear maps $\iota,\kappa:V\to \mathbb{R}^n$ such that
		\begin{equation}
			\iota(\Omega)\subseteq\Delta_n;\quad\kappa(\mathcal{E})\subseteq\Delta^*_n,
		\end{equation}
		and (iii) the probabilistic predictions are preserved, i.e.,
		\begin{equation}
			\braket{\varepsilon,s}=\kappa(\varepsilon)\cdot\iota(s),\quad\forall s\in\Omega,\varepsilon\in\mathcal{E}.
		\end{equation}
	\end{definition}
	The simplex $\Delta_n$ can be thought of as the space of probability distributions over a finite set of ontic states, and the hypercube $\Delta_n^*$ as the response functions over the finite set. This can therefore be thought of as a geometric representation of the ontological theory in which we wish to represent the GPT. This ontological theory is formally equivalent to the GPT representation of classical probability theory \cite{GPT2}.\\
	
	Undeniably, it is not always the case that an experiment has access to \emph{all} the states or effects in a GPT. In fact, in many cases the states and effects associated to an experiment will not even satisfy tomography. Moreover, if we consider nondeterministic sources as a way to prepare states in the experiment, then it can also be the case that  subnormalized states can be prepared whilst their normalized counterparts cannot. An \emph{accessible GPT fragment} of a GPT $(V,\braket{\cdot,\cdot},\Omega,\mathcal{E})$ was defined in Ref.~\cite{FRAGMENTS} to provide a description for such experiments. Formally it is a tuple $(I_{\Omega^F},I_{\mathcal{E}^F},\Omega^F,\mathcal{E}^F)$ such that $I_{\Omega^F}(\Omega^F)\subseteq\Omega$ and $I_{\mathcal{E}^F}^T(\mathcal{E}^F)\subseteq\mathcal{E}$, and where $I_{\Omega^F}$ and $I_{\mathcal{E}^F}$ are called \emph{inclusion maps}~\cite{FRAGMENTS}. Notice that there is no actual need to know the full GPT in order to define an accessible fragment -- the elements in $\Omega^F$ and $\mathcal{E}^F$ can be written in terms of the subspaces they span (which might not be dual to each other) rather than with respect to the full vector space $V$. Due to this possible mismatch between the spanned spaces, the inclusion processes $I_{\Omega^F}$ and $I_{\mathcal{E}^F}$ are needed to provide the prediction rule, that is,
	\begin{equation}
		p(\epsilon,s):=\braket{I_{\mathcal{E}^F}(\epsilon),I_{\Omega^F}(s)}, \quad \forall \epsilon\in\mathcal{E}_F,s\in\Omega_F.
	\end{equation}
	
	The notion of a simplex embedding can be straightforwardly imported to the accessible GPT fragment, and the failure of simplex embeddability for a fragment immediately implies the nonexistence of an embedding for the full GPT~\cite{CODE,FRAGMENTS}.  Importantly for us, it has been shown in Ref.~\cite{CODE} that one can test for simplex embeddability using a linear program. Moreover, one can also use this linear program to compute how robust a given scenario is to depolarizing noise. In the following, we briefly introduce this linear program and show how it can be easily adapted to also compute robustness to dephasing noise.
	\subsection{Modification of the linear program}
	We begin this section by summarising how the linear program from Ref.~\cite{CODE} works. Consider an  accessible GPT fragment $(I_{\Omega^F},I_{\mathcal{E}^F},\Omega^F,\mathcal{E}^F)$. The linear program takes as inputs $\Omega^F$ and $\mathcal{E}^F$, and first characterizes the facet inequalities of the positive cones of states/effects. Suppose that there are $n_{\Omega}\in \mathbb{N}$ of these for states and $n_{\mathcal{E}} \in \mathbb{N}$ of these for effects. The linear program then turns these collections of facet inequalities into the matrices $H_{\Omega}:\mathsf{LinSpan}[\Omega^F]\to \mathbb{R}^{n_\Omega}$ and $H_\mathcal{E}:\mathsf{LinSpan}[\mathcal{E}^F]\to \mathbb{R}^{n_\mathcal{E}}$ such that 
	\begin{equation}
		H_\Omega\cdot  s \geq_e0 \iff s=\sum_iq_is_i,\quad s_i\in\Omega^F,q_i\in\mathbb{R}^+,\,\forall i;
	\end{equation}
	\begin{equation}
		H_\mathcal{E}\cdot \varepsilon \geq_e0 \iff \varepsilon=\sum_iq_i\varepsilon_i,\quad \varepsilon_i\in\mathcal{E}^F,q_i\in\mathbb{R}^+,\,\forall i,
	\end{equation}
	where $\geq_e$ is entry-wise non-negativity. The code also characterizes the inclusion map $I_{\Omega^F}:\mathsf{LinSpan}[\Omega^F]\to V$ (and $I_{\mathcal{E}^F}:\mathsf{LinSpan}[\mathcal{E}^F]\to V$) which maps each state (effect) from the accessible GPT fragment to the smallest Euclidean vector space $V$ such that $\mathsf{LinSpan}[\Omega^F]\subseteq V$, $\mathsf{LinSpan}[\mathcal{E}^F]\subseteq V$ and with the dot product reproducing the probability rule. The code also takes as input a maximally mixed state $s_\mathcal{D}$, and characterizes the maximally depolarizing noise $\mathcal{D}$ from it.  Finally, it solves the following linear program:
	\begin{eqnarray}
		\text{min} & \quad & r \nonumber\\
		\text{s.t.} & \quad & rI^T_\mathcal{E}\cdot\mathcal{D}\cdot I_\Omega+(1-r)I_\mathcal{E}^T\cdot I_\Omega=H_\mathcal{E}^T\cdot\sigma\cdot H_\Omega,\\
		& & \sigma \geq_e 0\,. \nonumber
	\end{eqnarray}
	
	Now that we have summarised the main relevant aspect of the linear program of Ref.~\cite{CODE}, we can introduce the particular accessible GPT fragment employed in our work. A pure qubit state $\ket{\psi}$ rotated from state $\ket{0}$ by an arbitrary angle $\theta$ about the $Y$-axis can be represented in terms of an orthonormal operator basis as
	\begin{equation}
		\ket{\psi}\!\!\bra{\psi}=\frac12(\mathbb{1}+\sin\theta \hat{X}+\cos\theta \hat{Z}),
	\end{equation}
	where $\hat{X}$ (resp.~$\hat{Z}$) denotes the Pauli-X operator (resp.~Pauli-Z). We are assuming with no loss of generality that the plane in which our preparations and measurements live in the Bloch sphere is the $ZX$ plane. Since Hermitian operators in this plane can be represented by a real-valued, three-dimensional vector, our states and effects will have the following form:
	\begin{equation}
		\boldsymbol{\psi}:=\frac{1}{\sqrt2}\left(\begin{array}{c}
			1\\
			\sin\theta\\
			\cos\theta
		\end{array}\right);\quad
		\boldsymbol{\bar{\psi}}:=\frac{1}{\sqrt2}\left(\begin{array}{c}
			1\\
			-\sin\theta\\
			-\cos\theta
		\end{array}\right);\quad
		\boldsymbol{\phi}:=\frac{1}{\sqrt2}\left(\begin{array}{c}
			1\\
			\sin\theta\\
			-\cos\theta
		\end{array}\right);
	\end{equation}
	\begin{equation}
		\boldsymbol{\bar{\phi}}:=\frac{1}{\sqrt2}\left(\begin{array}{c}
			1\\
			-\sin\theta\\
			\cos\theta
		\end{array}\right);\quad
		\boldsymbol{E_{g_\psi}}:=\frac{1}{\sqrt2}\left(\begin{array}{c}
			1\\
			0\\
			1
		\end{array}\right);\quad
		\boldsymbol{E_{g_\psi}}:=\frac{1}{\sqrt2}\left(\begin{array}{c}
			1\\
			0\\
			-1
		\end{array}\right).
	\end{equation}
	One can also define a null vector and a unit vector,
	\begin{equation}
		\boldsymbol{0}:=\left(\begin{array}{c}
			0\\
			0\\
			0
		\end{array}\right);\quad	\boldsymbol{u}:=\left(\begin{array}{c}
			\sqrt2\\
			0\\
			0
		\end{array}\right), 
	\end{equation}
	and probabilities are given by taking the inner product between a preparation vector and an effect vector. This representation recovers all the expected statistics for this scenario. If we then define
	\begin{equation}
		\Omega^F:=\mathsf{Conv}\{\boldsymbol{\psi},\boldsymbol{\bar{\psi}},\boldsymbol{\phi},\boldsymbol{\bar{\phi}}\},
	\end{equation}
	\begin{equation}
		\mathcal{E}^F:=\mathsf{Conv}\{\boldsymbol{\psi},\boldsymbol{\bar{\psi}},\boldsymbol{\phi},\boldsymbol{\bar{\phi}},\boldsymbol{E_{g_\psi}},\boldsymbol{E_{g_\phi}},\boldsymbol{0},\boldsymbol{u}\},
	\end{equation}
	i.e., the convex hulls of the corresponding sets of vectors, then  $\mathcal{F}:=(\mathbb{1},\mathbb{1},\Omega^F,\mathcal{E}^F)$ is the accessible GPT 
	fragment associated with the prepare-and-measure scenario studied in this work. In this scenario, we can take the inclusion maps to be identities because the states and effects of the scenario are mutually tomographic.
	
	In quantum theory, a state $\rho$ will dephase in the Bloch sphere when $Z$ is the chosen basis, according to the dephasing channel $\mathcal{D}_Z$ defined by
	\begin{equation}
		\mathcal{D}_Z[\hat{\rho}]:=\sum_{i\in\{0,1\}}\frac12(\mathbb{1}+(-1)^i\hat{Z})\hat{\rho}\frac12(\mathbb{1}+(-1)^i\hat{Z})=\frac12(\mathbb{1}+\braket{Z}\hat{Z}).
	\end{equation}
	
	In the representation of the scenario as an accessible GPT fragment $\mathcal{F}$, this dephasing channel is represented by the linear map $\mathcal{D}_Z:\mathbb{R}^3\to\mathbb{R}^3$ given by 
	\begin{equation}
		\mathcal{D}_Z\circ\frac{1}{\sqrt2}\left(\begin{array}{c}
			1\\
			\braket{X}\\
			\braket{Z}
		\end{array}
		\right)=\frac{1}{\sqrt2}\left(\begin{array}{c}
			1\\
			0\\
			\braket{Z}
		\end{array}\right).
	\end{equation}
	
	More generally, for a general direction $\eta$ in the $ZX$ plane, we can define a dephasing map $\mathcal{D}_\eta$ in this representation. To start, define the projectors
	\begin{equation}
		\ket{+_\eta}\!\!\bra{+_\eta}=\frac12(\mathbb{1}+\cos\eta \hat{X}+\sin\eta \hat{Z}),\quad\ket{-_\eta}\!\!\bra{-_\eta}=\frac12(\mathbb{1}-\cos\eta \hat{X}-\sin\eta \hat{Z}),
	\end{equation}
	from which follows that
	\begin{equation}
		\mathcal{D}_\eta[\hat{\rho}]=\frac12(\mathbb{1}+(\braket{X}\cos^2\eta+\braket{Z}\cos\eta\sin\eta)\hat{X}+(\braket{X}\cos\eta\sin\eta+\braket{Z}\sin^2\eta)\hat{Z}).
	\end{equation}
	For the accessible GPT fragment $\mathcal{F}$, this action of the dephasing map is hence represented by
	\begin{equation}
		\mathcal{D}_\eta\circ\frac{1}{\sqrt2}\left(\begin{array}{c}
			1\\
			\braket{X}\\
			\braket{Z}
		\end{array}
		\right)=\frac{1}{\sqrt2}\left(\begin{array}{c}
			1\\
			\braket{X}\cos^2\eta+\braket{Z}\sin\eta\cos\eta\\
			\braket{X}\sin\eta\cos\eta+\braket{Z}\sin^2\eta
		\end{array}
		\right).
	\end{equation}
	This means that our GPT dephasing map in a generalized $\eta$ basis (in the $ZX$ plane) corresponds to the matrix
	\begin{equation}\label{dephasing}
		\mathcal{D}_\eta=
		\begin{pmatrix}
			1 & 0 & 0\\
			0 & \cos^2\eta & \cos\eta\sin\eta\\
			0 & \cos\eta\sin\eta & \sin^2\eta
		\end{pmatrix}.
	\end{equation}

	\begin{figure}[t]
		\centering
		\begin{tabular}{cc}
			\includegraphics[scale=0.5]{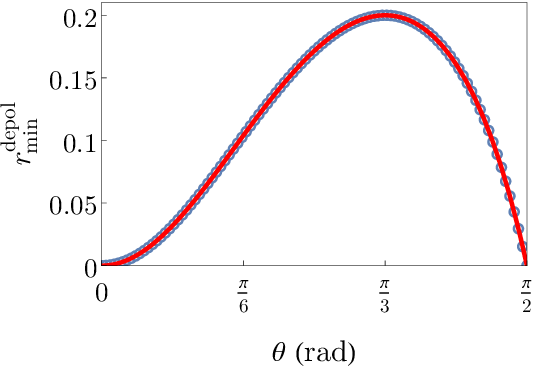}
			&
			\includegraphics[scale=0.5]{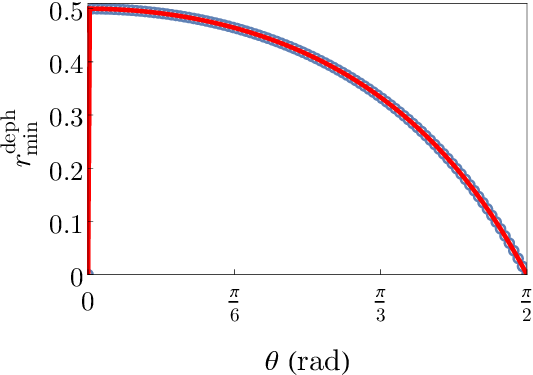}
		\end{tabular}
		\caption{To verify whether the original linear program (left) and its modification (right) survived the alterations, they were fed with the states and measurements of the original scenario (blue dots). Both obtained plots are compatible with the respective analytical plots (red curves).}
		\label{Fig7}
	\end{figure}
	
	The linear program from Ref.~\cite{CODE} takes as an input a set of states, a set of effects, and the unit vector from the accessible GPT fragment to be embedded. We modify the code to ask for an additional parameter $\eta$, and replace the occurrences of the depolarizing map with the matrix in Equation \ref{dephasing}. This modification makes sense for the particular scenario that we are interested in,  but adapting the program for analyzing robustness to dephasing for general quantum scenarios and for other GPTs is beyond the scope of this work. Both the original linear program and the modified version also were altered to include the robustness as the first element of their string of outputs ( at first, Ref.~\cite{CODE} would output only the list of epistemic states and response functions of the obtained noncontextual model), in order to make the plots easier. As a verification of whether the codes were functioning as expected, we demonstrated that both the original linear program and its modification managed to recover the analytical results from the main text, as shown in Figure ~\ref{Fig7}.
	\section{MESD with dephasing noise and noisy and rotated discriminating measurement}\label{app4}
	
	Here we introduce the plots obtained via applying the modified linear program introduced in Sec.~\ref{app3} to various scenarios (including different amounts and types of noise).
	
	The first case we study is one where we rotate measurement $M_g$ by a parameter $\alpha$ with respect to the $Z$ axis, in order to investigate how the robustness to dephasing behaves in these related scenarios with a more coherent $M_g$.  The plot is displayed in Fig. \ref{Fig8}. As expected, the plot interpolates between the case in which $M_g$ lies aligned to the $Z$ axis (Fig. 3 of the main text) and the $X$ axis (Fig. 4 of the main text). The only circumstances in which robustness is null are when the measurement $M_g$ coincides with one of the other measurements, i.e., when $\alpha=\theta$. Furthermore, the plot shows clearly that $\alpha=\frac{\pi}{2}$ is the only circumstance in which the robustness saturates to $1$.
	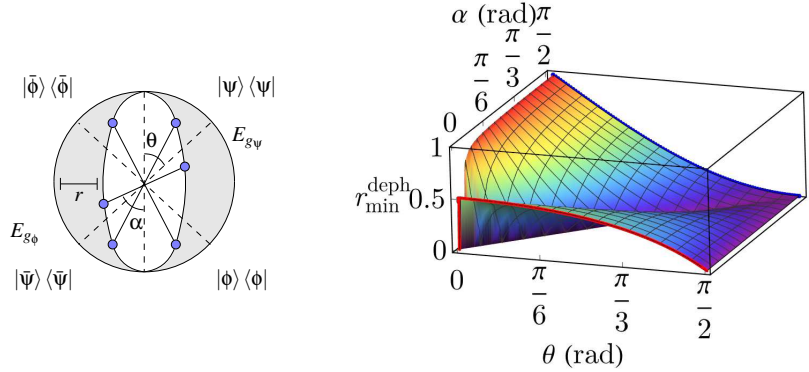
\begin{figure}[h]
		\centering
		\begin{tabular}{cc}
			\adjustbox{scale=0.8,valign=m}{\begin{tikzpicture}
	\begin{pgfonlayer}{nodelayer}
		\node [style=hadamard] (0) at (4.6, 2.25) {};
		\node [style=hadamard] (1) at (3.25, 1.63) {};
		\node [style=none] (2) at (5.65, 2.7125) {$E_{g_\psi}$};
		\node [style=none] (3) at (1.95, 1.1125) {$E_{g_\phi}$};
		\node [style=none] (4) at (3.925, 3.5) {};
		\node [style=none] (5) at (2.425, 2) {};
		\node [style=none] (6) at (3.925, 0.5) {};
		\node [style=none] (7) at (5.425, 2) {};
		\node [style=hadamard] (8) at (4.45, 2.975) {};
		\node [style=hadamard] (9) at (3.4, 0.955) {};
		\node [style=none] (10) at (5.625, 3.575) {$\ket{\psi}\bra{\psi}$};
		\node [style=hadamard] (11) at (3.4, 2.975) {};
		\node [style=hadamard] (12) at (4.45, 0.955) {};
		\node [style=none] (13) at (2.325, 3.575) {$\ket{\bar{\phi}}\bra{\bar{\phi}}$};
		\node [style=none] (14) at (2.25, 0.375) {$\ket{\bar{\psi}}\bra{\bar{\psi}}$};
		\node [style=none] (15) at (5.55, 0.375) {$\ket{\phi}\bra{\phi}$};
		\node [style=none] (16) at (2.525, 1.9625) {};
		\node [style=none] (17) at (3.15, 1.9625) {};
		\node [style=none] (18) at (2.84, 1.762) {$r$};
		\node [style=none] (19) at (5.025, 2.975) {};
		\node [style=none] (20) at (2.835, 0.955) {};
		\node [style=none] (21) at (5.025, 0.955) {};
		\node [style=none] (22) at (2.835, 2.975) {};
		\node [style=none] (23) at (3.925, 2.4625) {};
		\node [style=none] (24) at (4.265, 2.2685) {};
		\node [style=none] (25) at (4.05, 2.7) {$\theta$};
		\node [style=none] (26) at (3.915, 1.531) {};
		\node [style=none] (27) at (3.575, 1.75) {};
		\node [style=none] (28) at (3.79, 1.2935) {$\alpha$};
	\end{pgfonlayer}
	\begin{pgfonlayer}{edgelayer}
		\draw (1) to (0);
		\draw [dashed] (4.center) to (6.center);
		\draw [bend left=45] (5.center) to (4.center);
		\draw [bend left=45] (4.center) to (7.center);
		\draw [bend left=45] (7.center) to (6.center);
		\draw [bend left=45] (6.center) to (5.center);
		\draw (9) to (8);
		\draw (12) to (11);
		\draw [bend left] (4.center) to (8);
		\draw [bend left=15] (8) to (12);
		\draw [bend left] (12) to (6.center);
		\draw [bend left] (6.center) to (9);
		\draw [bend left=15] (9) to (11);
		\draw [bend left] (11) to (4.center);
		\draw [{|-|}] (17.center) to (16.center);
		\draw [dashed] (19.center) to (20.center);
		\draw [dashed] (21.center) to (22.center);
		\draw [bend left] (23.center) to (24.center);
		\draw [bend left] (26.center) to (27.center);
	\end{pgfonlayer}

 \begin{pgfonlayer}{background}
		\fill[gray!20] (4) to[bend right=45] (5) to[bend right=45] (6) to[bend right=45] (7) to[bend right=45] (4);
		\fill [white] (6.center) to[bend left=40] (9.center) to[bend left=15] (11.center) to[bend left=40] (4.center) to[bend left=40]  (8.center) to[bend left=15] (12.center) to[bend left=40] (6.center);
	\end{pgfonlayer}
\end{tikzpicture}}
			&
			\adjustbox{scale=0.7,valign=m}{\includegraphics[scale=1.2]{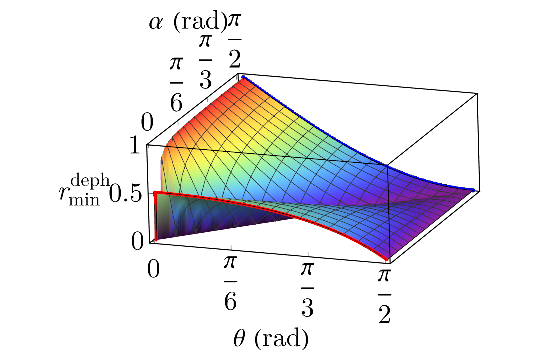}}
		\end{tabular}
		\caption{Left: representation of the action of dephasing noise on the plane in which measurements live, now with measurement $M_g$ rotated by an angle $\alpha$; Right: contextual robustness to dephasing as a function of both the angle $\theta$ between the prepared states and the $Z$ axis the angle $\alpha$ by which $M_g$ is rotated. The red curve represents the plot from Fig. 3 of the main text and the blue curve represents the plot from Fig.~4.}
		\label{Fig8}
	\end{figure}
	
	The second case of study is the scenario where the measurement $M_g$ has been affected by some amount $p$ of depolarizing noise prior to the assessment of the robustness of the scenario to dephasing (i.e., before computing $r_{\text{min}}$).  The scheme for this scenario is given in Figure~\ref{Fig9}, along with a plot for the robustness to dephasing as a function of both $\theta$ and $p$ (the impurity of the discriminating measurements).
	\begin{figure}[h!]
		\centering       
		\begin{tabular}{cc}
			\adjustbox{valign=m,scale=0.8}{\begin{tikzpicture}
	\begin{pgfonlayer}{nodelayer}
		\node [style=hadamard] (2) at (3.925, 3.3) {};
		\node [style=hadamard] (3) at (3.925, 0.7) {};
		\node [style=none] (4) at (3.925, 3.925) {$E_{g_\psi}$};
		\node [style=none] (5) at (3.925, 0) {$E_{g_\phi}$};
		\node [style=none] (6) at (3.925, 3.5) {};
		\node [style=none] (7) at (2.425, 2) {};
		\node [style=none] (8) at (3.925, 0.5) {};
		\node [style=none] (9) at (5.425, 2) {};
		\node [style=hadamard] (10) at (4.45, 2.975) {};
		\node [style=hadamard] (11) at (3.4, 0.955) {};
		\node [style=none] (12) at (5.625, 3.275) {$\ket{\psi}\bra{\psi}$};
		\node [style=hadamard] (13) at (3.4, 2.975) {};
		\node [style=hadamard] (14) at (4.45, 0.955) {};
		\node [style=none] (15) at (2.325, 3.275) {$\ket{\bar{\phi}}\bra{\bar{\phi}}$};
		\node [style=none] (16) at (2.25, 0.675) {$\ket{\bar{\psi}}\bra{\bar{\psi}}$};
		\node [style=none] (17) at (5.55, 0.675) {$\ket{\phi}\bra{\phi}$};
		\node [style=none] (18) at (2.525, 1.9625) {};
		\node [style=none] (19) at (3.175, 1.9625) {};
		\node [style=none] (20) at (2.84, 1.762) {$r$};
		\node [style=none] (21) at (5.025, 2.975) {};
		\node [style=none] (22) at (2.835, 0.955) {};
		\node [style=none] (23) at (5.025, 0.955) {};
		\node [style=none] (24) at (2.835, 2.975) {};
		\node [style=none] (25) at (3.925, 2.4625) {};
		\node [style=none] (26) at (4.265, 2.2685) {};
		\node [style=none] (27) at (4.05, 2.7) {$\theta$};
		\node [style=none] (28) at (3.925, 3.475) {};
		\node [style=none] (29) at (3.925, 3.375) {};
		\node [style=none] (30) at (4.975, 3.425) {$p$};
		\node [style=none] (31) at (4.75, 3.475) {};
		\node [style=none] (32) at (4.75, 3.375) {};
	\end{pgfonlayer}
	\begin{pgfonlayer}{edgelayer}
		\draw (3) to (2);
		\draw [bend left=45] (7.center) to (6.center);
		\draw [bend left=45] (6.center) to (9.center);
		\draw [bend left=45] (9.center) to (8.center);
		\draw [bend left=45] (8.center) to (7.center);
		\draw (11) to (10);
		\draw (14) to (13);
		\draw [bend left=40] (6) to (10);
		\draw [bend left=15] (10) to (14);
		\draw [bend left=40] (14) to (8);
		\draw [bend left=40] (8) to (11);
		\draw [bend left=15] (11) to (13);
		\draw [bend left=40] (13) to (6);
		\draw [{|-|}] (19.center) to (18.center);
		\draw [dashed] (21.center) to (22.center);
		\draw [dashed] (23.center) to (24.center);
		\draw [bend left] (25.center) to (26.center);
		\draw [dashed] (31) to (28.center);
		\draw [dashed] (29.center) to (32);
		\draw [{|-|}] (32.center) to (31.center);
	\end{pgfonlayer}

	\begin{pgfonlayer}{background}
		\fill[gray!20] (6) to[bend right=45] (7) to[bend right=45] (8) to[bend right=45] (9) to[bend right=45] (6);
		\fill [white] (6.center) to[bend left=40] (10.center) to[bend left=15] (14.center) to[bend left=40] (8.center) to[bend left=40]  (11.center) to[bend left=15] (13.center) to[bend left=40] (6.center);
	\end{pgfonlayer}
\end{tikzpicture}} & \hspace{30pt}\adjustbox{valign=m}{\includegraphics[scale=0.65]{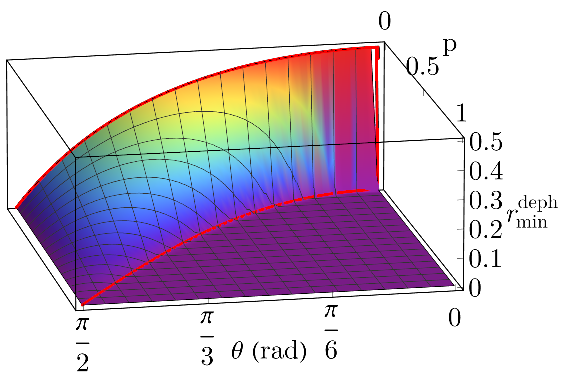}} 
		\end{tabular}
		\caption{Left: representation of the action of dephasing noise on the plane in which measurements live, now with discriminating measurement undergoing depolarizing noise by a factor $p$; Right: contextual robustness to dephasing as a function of both the angle between the prepared states and the $Z$ axis and the noise added to the discriminating measurement. Red curves represent the plot from Fig.~3 of the main text and the equality from Eq.~\ref{estimation}.}
		\label{Fig9}
	\end{figure}
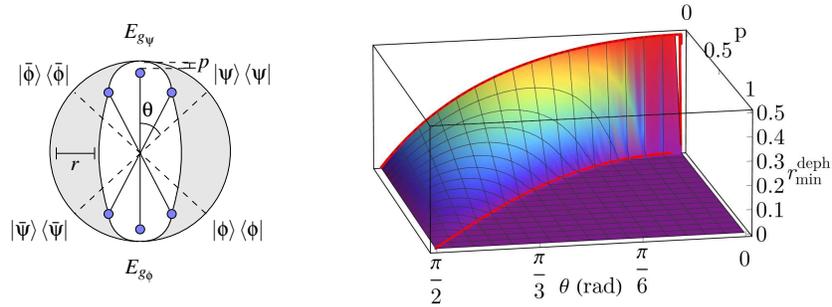
	
	Notice that the section of the plot where $p=0$ corresponds to the plot from Figure~3 of the main text, that is, the dephased scenario with measurement $M_g$ aligned with the $Z$ axis. There is a clear relation between $p$ and $\theta$ from which no contextuality can be proven, and numerically it coincides with the equation
	\begin{equation}\label{estimation}
		p= 1-\cos\theta.
	\end{equation}
	Eq.~\eqref{estimation} hence provides the maximum noise one can add to the measurement $M_g$ so there is still a proof of contextuality when the other measurements undergo dephasing noise. Geometrically, $p\geq1-\cos\theta$ represents a measurement $M_g$ where the depolarizing noise $p$, has made it such that its effects become merely convex combinations of the other effects in the scenario. Because one cannot prove contextuality with just the four preparations and their corresponding effects alone, scenarios with $p\geq1-\cos\theta$ admit of a noncontextual ontological model. If trace distance is the quantifier of coherence employed, as per Equation~1 in the main text, then the following inequality
	\begin{equation}\label{resource}
		C(\ket{\psi}\!\!\bra{\psi})>\sqrt{p(2-p)}
	\end{equation}
	tells us how much coherence the prepared states and measurements must start with so that the scenario can still prove contextuality despite the noise. Nevertheless, the scenario becomes more and more sensitive to dephasing noise as the impurity $p$ increases, such that even for relatively small values of $p$ the maximum robustness achieved decreases considerably. Notice still that for undisturbed measurement $M_g$ ($p=0$), inequality \ref{resource} agrees with Result $1$ of the main text: proofs of contextuality will be achieved as long $C(\ket{\psi}\!\!\bra{\psi})>0$.

	\bigskip 
	\twocolumngrid
	\bibliographystyle{apsrev4-2}
	\bibliographystyle{apsrev4-2}

\end{document}